\documentclass [11pt]{article}

\usepackage{fullpage,wrapfig}
\usepackage{amsmath, amsthm, amssymb, algorithm,thmtools,algpseudocode,enumerate,caption,framed}
\usepackage{paralist, sidecap}

\usepackage[usenames,dvipsnames,svgnames,table]{xcolor}
\definecolor{darkgreen}{rgb}{0.0,0,0.9}

\usepackage[colorlinks=true,pdfpagemode=UseNone,citecolor=Blue,linkcolor=Red,urlcolor=BrickRed,
pagebackref]{hyperref}

\usepackage{tikz}

\newcommand{\sat}{\texttt{SAT}}
\newcommand{\usc}{\texttt{USC}}
\newcommand{\start}{\texttt{start}}
\newcommand{\prevv}{\texttt{prev}'}
\newcommand{\prev}{\texttt{prev}}
\newcommand{\curr}{\texttt{curr}}
\newcommand{\currr}{\texttt{curr}'}
\newcommand{\en}{\texttt{end}}

\makeatletter
\newcommand{\setword}[2]{%
  \phantomsection
  #1\def\@currentlabel{\unexpanded{#1}}\label{#2}%
}
\makeatother

\numberwithin{equation}{section}

\newtheorem{theorem}{Theorem}[section]

\newtheorem{lemma}{Lemma}[section]

\title{Unique Set Cover on Unit Disks and Unit Squares}

\author{
Saeed Mehrabi
\thanks{Cheriton School of Computer Science, University of Waterloo, Canada.
Email: \protect\url{smehrabi@uwaterloo.ca}}}

\date{}

\begin{document}

\maketitle

\begin{abstract}
We study the \textsc{Unique Set Cover} problem on unit disks and unit squares. For a given set $P$ of $n$ points and a set $D$ of $m$ geometric objects both in the plane, the objective of the \textsc{Unique Set Cover} problem is to select a subset $D'\subseteq D$ of objects such that every point in $P$ is covered by at least one object in $D'$ and the number of points covered uniquely is maximized, where a point is \emph{covered uniquely} if the point is covered by exactly one object in $D'$. In this paper, \begin{inparaenum}[(i)] \item we show that the \textsc{Unique Set Cover} is \textsc{NP}-hard on both unit disks and unit squares, and \item we give a PTAS for this problem on unit squares by applying the \emph{mod-one} approach of Chan and Hu (Comput. Geom. 48(5), 2015). \end{inparaenum}
\end{abstract}

\section{Introduction}
\label{sec:introduction}
Consider a set $P$ of $n$ points and a set $D$ of $m$ geometric objects both in the plane. For a subset $S\subseteq D$ of objects, we say that a point $p\subseteq P$ is \emph{covered uniquely} by $S$ if there is exactly one object in $S$ that covers $p$. In the \emph{\textsc{Unique Set Cover}} problem, the objective is to compute a subset $S\subseteq D$ of objects so as to cover all points in $P$ and to maximize the number of points covered uniquely by $S$.

There is a slightly different problem, called \textsc{Unique Cover}, where the input is the same as the \textsc{Unique Set Cover} problem and the objective is to compute a subset $S\subseteq D$ that maximizes the number of points covered uniquely; note that not all the points in $P$ are required to be covered in an instance of the \textsc{Unique Cover} problem. Both \textsc{Unique Set Cover} and \textsc{Unique Cover} belong to the larger class of the \textsc{Unit Cover} problem in which we are given the same input and the objective is to compute a minimum-cardinality subset $S\subseteq D$ so as to cover all points in $P$. Indeed, \textsc{Unique Set Cover} combines the objectives of the \textsc{Unique Cover} and \textsc{Unit Cover} problems.

In this paper, we are interested in the \textsc{Unique Set Cover} problem on unit disks and unit squares. Formally, given a set $P$ of $n$ points and a set $D$ of $m$ unit disks (resp., unit squares) both in the plane, the objective of the \textsc{Unique Disk} (resp., \textsc{Square}) \textsc{Set Cover} problem is to find a subset $S\subseteq D$ that covers all points in $P$ and that maximizes the number of points covered uniquely. Figure~\ref{fig:problemInstance} shows an instance of, e.g., the \textsc{Unique Disk Set Cover} problem and an optimal solution for this instance.

The \textsc{Unique Cover} problem was first studied by Erlebach and van Leeuwen~\cite{ErlebachL2008} on unit disks and is motivated by its applications in wireless communication networks, where the broadcasting range of equivalent base stations have been frequently modelled as unit disks. Providers of wireless networks often consider having a number of base stations to service their customers. However, if a customer receives signals from too many base stations, then the customer may receive no service at all due to the resulting interference. As such, each customer is ideally serviced by exactly one base station and we want such a service to be provided to as many customers as possible. Our motivation for studying the \textsc{Unique Set Cover} problem is that what if in addition to providing service to as many customers as possible by exactly one base station, we still want to do service \emph{all} the customers as well.

\begin{figure}[t]
\centering
\includegraphics[width=.60\textwidth]{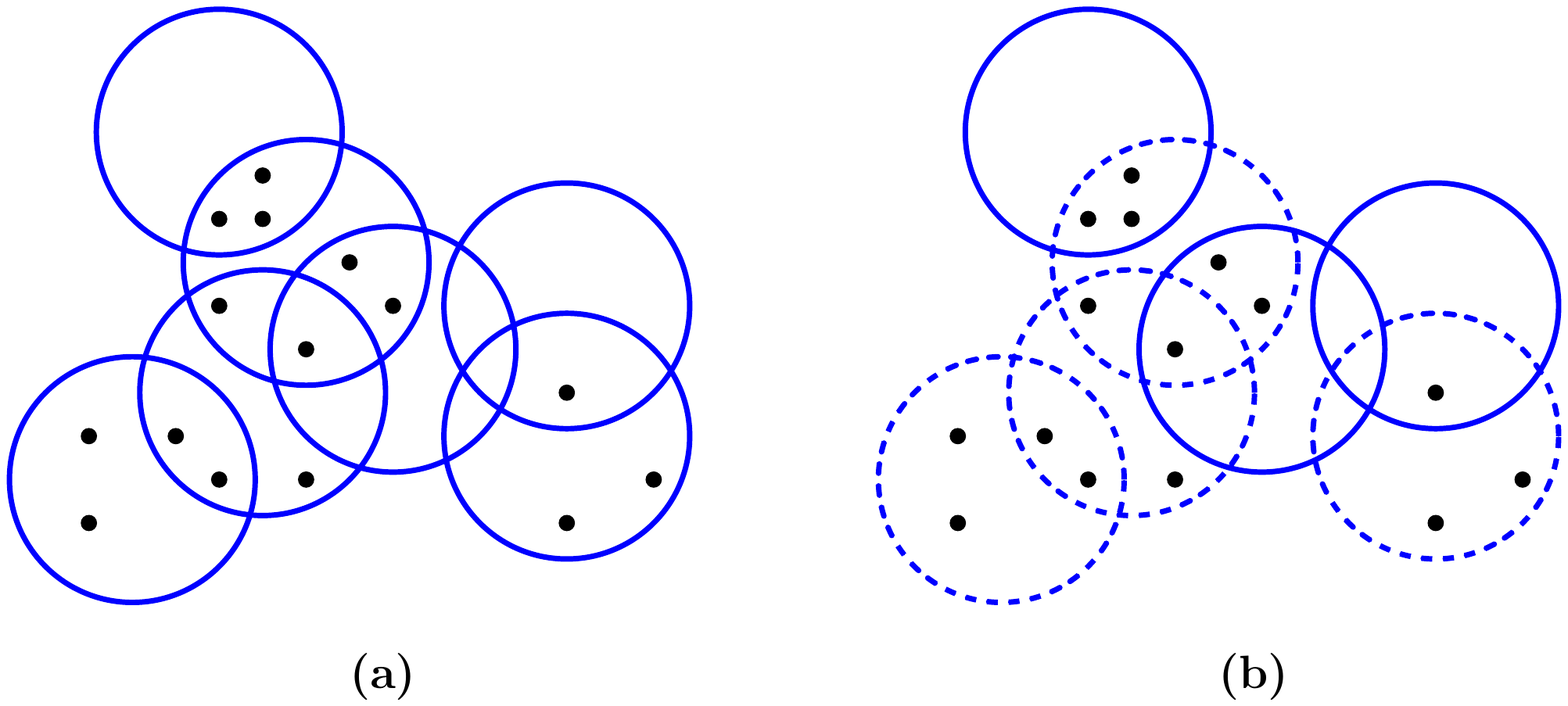}
\caption{{\small (a) An instance of the \textsc{Unique Disk Set Cover} problem with $n=15$ and $m=7$, and (b) an optimal solution, where 11 points are covered uniquely by the 4 dashed disks.}}
\label{fig:problemInstance}%
\end{figure}

\paragraph{Related Work.} The \textsc{Unit Cover} problem is a well-known \textsc{NP}-hard problem on unit disks and unit squares. Mustafa and Ray~\cite{MustafaR2010} gave the first PTAS for the \textsc{Unit Cover} problem on both unit disks and unit squares using a \emph{local search technique}. This technique was also discovered independently by Chan and Har-Peled~\cite{ChanHP2012} who gave the first PTAS for the maximum independent set problem on pseudo-disks in the plane. Demaine et al.~\cite{DemaineFHS2008} introduced the non-geometric variant of the \textsc{Unique Cover} problem and gave a polynomial-time $O(\log n)$-approximation algorithm for the problem, where $n$ is the number of elements of the universe in the corresponding set system. Erlebach and van Leeuwen~\cite{ErlebachL2008}, who were first to study the geometric version of the \textsc{Unique Cover} problem, showed that \textsc{Unique Cover} problem is \textsc{NP}-hard on both unit disks and unit squares (see also~\cite{ErikJVLeeuwen2009}). 

By combining dynamic programming and the shifting strategy of Hochbaum and Maass~\cite{HochbaumM1985}, Erlebach and van Leeuwen~\cite{ErlebachL2010} gave the first PTAS for the weighted version of the \textsc{Unit Cover} problem on unit squares, where each unit square in $D$ is associated with a positive value as \emph{weight} and the objective is to cover the points in $P$ so as to minimize the total weight of the unit squares selected.\footnote{The local search technique of Mustafa and Ray~\cite{MustafaR2010} does not apply to the weighted version of these problems.} This approach was also used by Ito et al.~\cite{ItoNOOUUU2016} who gave a PTAS for the \textsc{Unique Cover} problem on unit squares. However, this technique does not seem to work for unit disks. In fact, Ito et al.~\cite{ItoNOOUUU2014} used this approach to give a polynomial-time approximation algorithm for the \textsc{Unique Cover} problem on unit disks with approximation factor $\sigma<4.3095+\epsilon$, where $\epsilon>0$ is any fixed constant. Moreover, this approach involves sophisticated dynamic programming. Finally, by introducing a \emph{mod-one transformation}, Chan and Hu~\cite{ChanH2015} gave a PTAS for the \textsc{Red-Blue Unit-Square Cover} problem, which is defined as follows: given a red point set $R$, a blue point set $B$ and a set of unit squares in the plane, the objective is to select a subset of the unit squares so as to cover all the blue points while minimizing the number of red points covered.

\paragraph{Our Results.} In this paper, we first show that both \textsc{Unique Disk Set Cover} and \textsc{Unique Square Set Cover} are \textsc{NP}-hard. We note that the \textsc{Unique Cover} problem is shown to be \textsc{NP}-hard on unit disks and unit squares~\cite{ErlebachL2008}. However, their hardness, which is based on a reduction from the Independent Set problem on planar graphs of maximum degree 3, does not apply to proving the hardness of the \textsc{Unique Set Cover} problem mainly because all points in $P$ are required to be covered in an instance of the \textsc{Unique Set Cover} problem. We instead show a reduction from a variant of the planar 3SAT problem to prove the \textsc{NP}-hardness of \textsc{Unique Set Cover} on both unit disks and unit squares.

As our second result, we show that the \emph{mod-one} approach of Chan and Hu~\cite{ChanH2015} provides a PTAS for the \textsc{Unique Square Set Cover} problem. The only difference is that instead of storing 4-tuples of unit squares when solving the dynamic programming, we store 6-tuples of unit squares and show that they suffice to capture the necessary information (we will discuss this approach in more details in Section~\ref{sec:ptas}). See Table~\ref{tbl:ourTable} for a summary of previous and new results.

We first prove the \textsc{NP}-hardness of the problems in Section~\ref{sec:npHardness} and will then give a PTAS for the problem on unit squares in Section~\ref{sec:ptas}. Finally, we conclude the paper with a discussion on open problems in Section~\ref{sec:conclusion}.

\begin{table*}[t]
\centering
\begin{tabular}{ |p{4cm}||p{4.5cm}|p{6cm}| }
 \hline
 Problem & Unit Squares & Unit Disks\\
 \hline
 \textsc{Unit Cover} & \textsc{NP}-hard, PTAS~\cite{MustafaR2010} & \textsc{NP}-hard, PTAS~\cite{MustafaR2010}\\
 \textsc{Unique Cover} & \textsc{NP}-hard~\cite{ErlebachL2008}, PTAS~\cite{ItoNOOUUU2016} & \textsc{NP}-hard~\cite{ErlebachL2008}, 4.31-approximation~\cite{ItoNOOUUU2014}\\
 \textsc{Unique Set Cover} & {\bf \textsc{NP}-hard} [Theorem~\ref{thm:uniqueSetCoverDisksHardness}] & {\bf \textsc{NP}-hard} [Theorem~\ref{thm:unitSquaresHardness}]\\
 & {\bf PTAS} [Theorem~\ref{thm:shiftingStrategy}] & \\
 \hline
\end{tabular}
\caption{{\small A summary of previous and new results; the new results are shown in bold.}}
\label{tbl:ourTable}
\end{table*}

\section{Hardness}
\label{sec:hardness}
In this section, we first show that the \textsc{Unique Disk Set Cover} is \textsc{NP}-hard; the \textsc{NP}-hardness of the \textsc{Unique Square Set Cover} is proved analogously and we will discuss it at the end of this section.\\

\noindent\textsc{Unique Disk Set Cover}

\noindent{\bf Input.} A set $P$ of $n$ points, a set $D$ of $m$ unit disks, and an integer $k>0$.

\noindent{\bf Output.} Does there exist a set $S\subseteq D$ that covers all points in $P$ and that covers at least $k$ points uniquely?\\

To prove the hardness of the \textsc{Unique Disk Set Cover}, we show a reduction from the \textsc{Planar Variable Restricted 3SAT} (\textsc{Planar VR3SAT}, for short) problem. \textsc{Planar VR3SAT} is a constrained version of \textsc{3SAT} in which each variable can appear in at most three clauses and the corresponding \emph{variable-clause graph} is planar. Efrat et al.~\cite{EfratEK2007} showed that \textsc{Planar VR3SAT} is \textsc{NP}-hard. Let $I_\sat$ be an instance of \textsc{Planar VR3SAT} with $K$ clauses $C_1, C_2, \ldots, C_K$ and $N$ variables $X_1, X_2, \ldots, X_N$; we denote the two literals of a variable $X_i$ by $x_i$ and $\overline{x_i}$. We construct an instance $I_\usc$ of the \textsc{Unique Disk Set Cover} problem such that $I_\usc$ has a solution with at least $c\cdot (K+1)$ points covered unqiuely, for some $c$ that we will determine its value later, if and only if $I_\sat$ is satisfiable. Given $I_\sat$, we first construct the variable-clause graph $G$ of $I_\sat$ in the non-crossing comb-shape form of Knuth and Raghunathan~\cite{KnuthR1992}. Without loss of generality, we assume that the variable vertices lie on a vertical line and the clause vertices are connected from left or right of that line; see Figure~\ref{fig:3SatGraph} for an illustration. Note that each variable appears in at most three clauses.

\begin{wrapfigure}{r}{0.4\textwidth}
\centering
\vspace{-10pt}
\includegraphics[scale=.55]{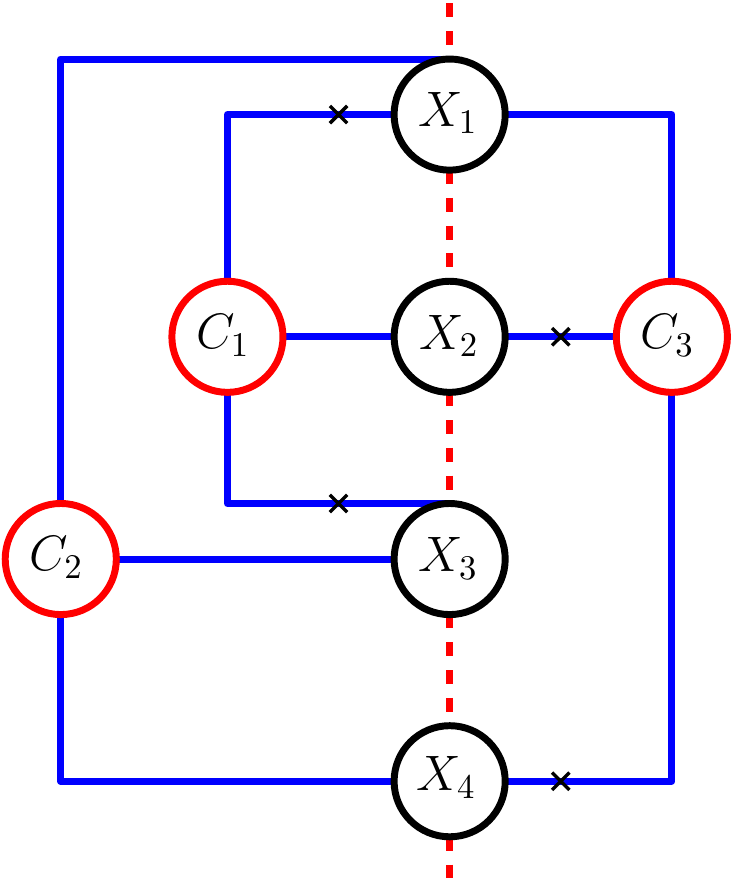}
\caption{{\small An instance of the \textsc{Planar VR3SAT} problem in the comb-shape form of Knuth and Raghunathan~\cite{KnuthR1992}. Crosses on the edges indicate negations; for example, $C_1=(\overline{x_1}\lor x_2 \lor \overline{x_3})$.}}
\label{fig:3SatGraph}%
\vspace{-10pt}
\end{wrapfigure}

\begin{figure}[t]
\centering
\includegraphics[width=.40\textwidth]{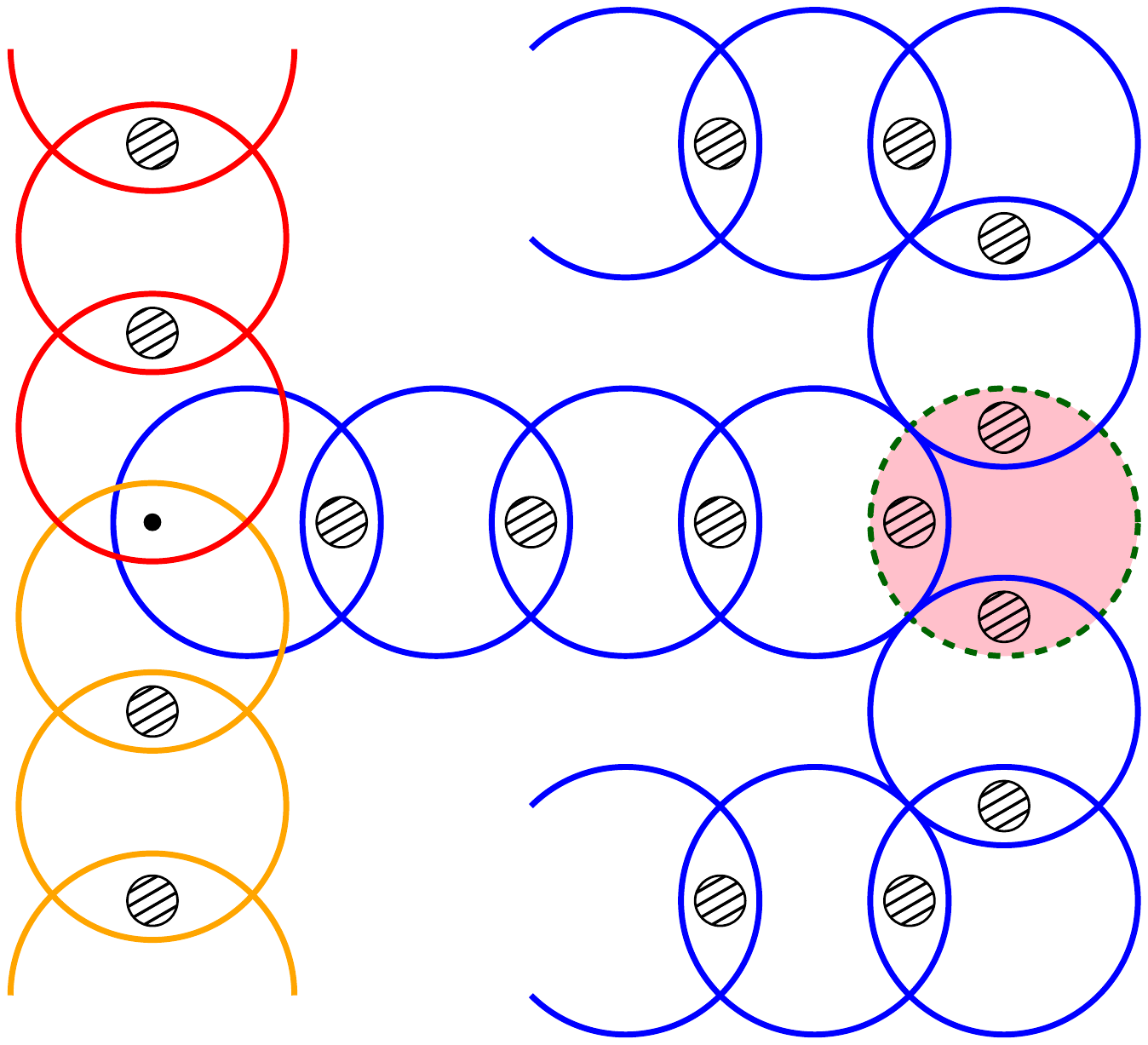}
\caption{{\small An illustration of the gadgets used in our reduction. {\bf Right.} The filled (pink) unit disk indicates a variable unit disk and each small (rising) shaded circle indicates a cloud consisting of $K+1$ points. Each variable unit disk shares three clouds with the start disks of the three wires that connect the variable unit disk to the clauses in which it appears. {\bf Left.} A clause gadget determined by the non-empty intersection of the last three unit disks of the three wires arriving from the literals of this clause, and the corresponding clause point.}}
\label{fig:variableGadget}%
\end{figure}

\paragraph{Gadgets.} For each variable $X_i\in I_\sat$, we replace the corresponding variable vertex in $G$ with a single unit disk containing three \emph{groups of points} each of which consists of $K+1$ points (where $K$ is the number of clauses in $I_\sat$). We call each such groups of points a \emph{cloud}. If a literal of $X_i$ appears in a clause, then the variable gadget of $X_i$ is connected to the corresponding clause gadget by a chain of unit disks such that \begin{inparaenum}[(i)] \item the first unit disk in the chain covers exactly one of the clouds in the variable unit disk, and \item every two consecutive unit disks in the chain share a cloud. \end{inparaenum} We call such a chain of unit disks a \emph{wire}; see Figure~\ref{fig:variableGadget}(Right) for an illustration. A cloud shared between two unit disks in a wire has also $K+1$ points. We call the unit disk of a wire that shares a cloud with the variable unit disk a \emph{start disk}. For the clause gadget, where three wires meet, we make the last three unit disks (each of which arriving from one of the wires) to have a non-empty intersection region in which we insert one single point; see Figure~\ref{fig:variableGadget}(Left). We call this point a \emph{clause point}.

Consider the cloud shared between a variable unit disk and a start disk. If all the clouds in the corresponding wire are covered uniquely, then exactly one of the variable unit disk and the start disk must be selected to cover the cloud shared between them. Consequently, depending on whether the variable unit disk is selected, we can decide whether the clause point of the clause gadget on the other end of the wire is covered by the last unit disk of this wire. That is, we can create two ways of covering the clouds of the wire uniquely, one of which will not be able to cover the corresponding clause point. It is clear from the variable unit disk shown in Figure~\ref{fig:variableGadget} that if the variable unit disk is selected (resp., is not selected), then the clause point is covered (resp., is not covered) by the last unit disk of the corresponding wire. We remark that this is always doable, even for the wires that require one bend, by adjusting the number of unit disks in the wire. See the top wire shown in Figure~\ref{fig:completeFigure} for an example. We set a variable to true or false depending on whether its corresponding variable unit disk is selected or not to cover the clouds that it contains.

\begin{wrapfigure}{l}{0.4\textwidth}
\centering
\includegraphics[width=.35\textwidth]{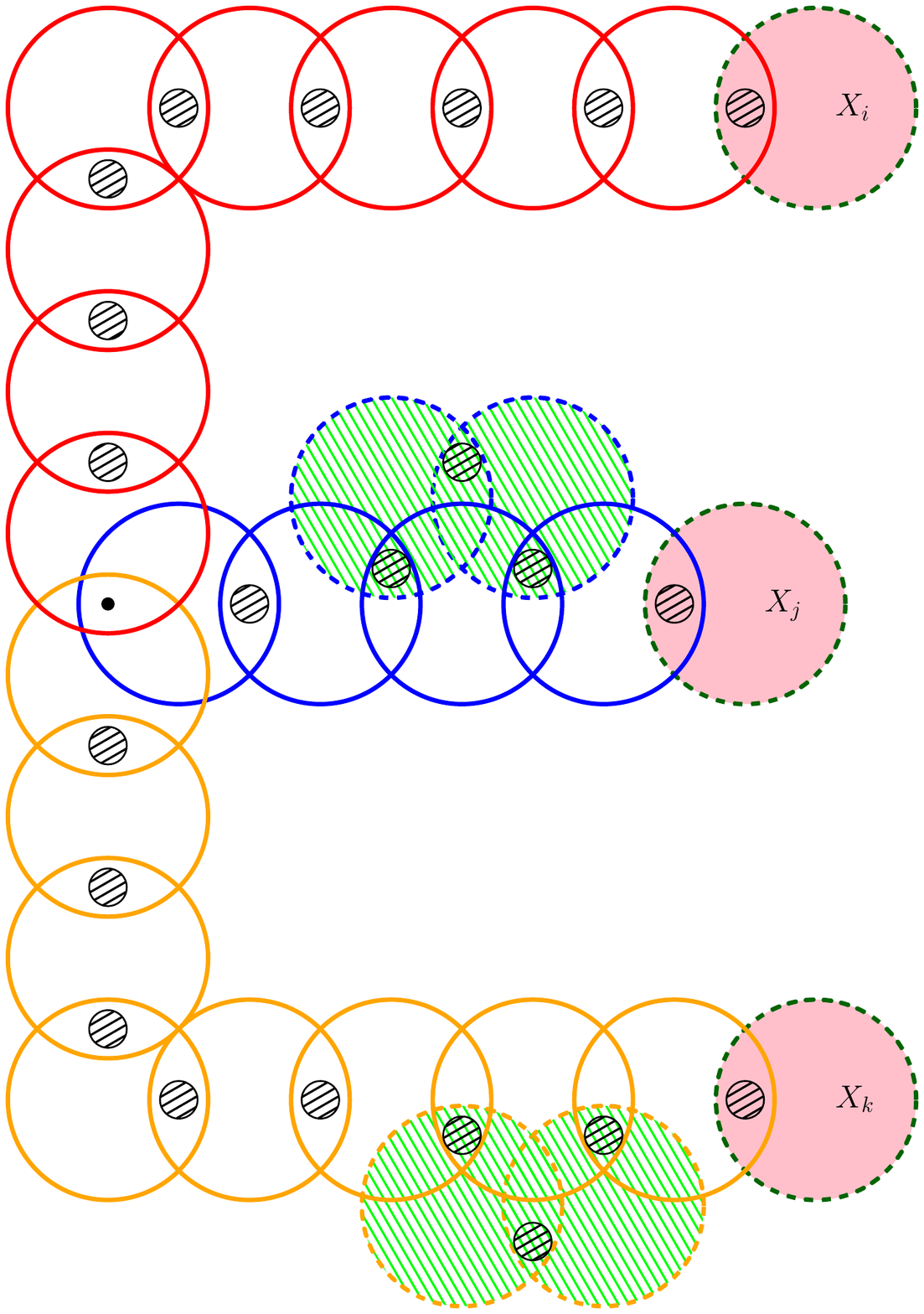}
\caption{{\small A complete illustration of a clause gadget $C=(x_i\lor \overline{x_j}\lor \overline{x_k})$ and its corresponding variable gadgets. The existence of a negation pair indicates that literal $\overline{x_j}$ (or literal $\overline{x_k}$) appears in $C$.}}
\label{fig:completeFigure}%
\end{wrapfigure}

Finally, we need a gadget for negation literals. Suppose that the literal $\overline{x_i}$ appears in a clause, for some variable $X_i$ in $I_\sat$. To indicate the negation literal, we add two additional unit disks besides the wire connecting the variable unit disk to the corresponding clause gadget; see the two (falling and green) shaded unit disks in the middle wire (i.e., the wire corresponding to $X_j$) shown in Figure~\ref{fig:completeFigure}. We call these two newly-added unit disks the \emph{negation pair}. Observe that the negation pair share a cloud and each of them covers one of the previously existed clouds in the current wire. The construction of the negation pair ensures that if the variable unit disk is selected (resp., is not selected), i.e., it is set to true (resp., false), then the clause point is not covered (resp., is covered). The same gadget can be used for the negation literals whose corresponding wires have a bend; see for instance the lowest wire shown in Figure~\ref{fig:completeFigure}.

\paragraph{Construction Details.} Clearly, our construction allows the wires to be connected to a variable unit disk from both left and right of the variable unit disk; we just need to move the corresponding clouds inside the variable unit disk accordingly. By scaling and making the drawing of the wires consistent with the edges of $G$, we can ensure that the unit disks of different wires will never intersect. Note that we can have the scaling so as to increase the number of unit disks in any wire by only a constant factor. Hence, by this and from the construction, we have that the number of disks in the instance $I_\usc$ is polynomial in terms of $K$ and $N$. To see the value of $n$, the number of points in $I_\usc$, we set $c$ to the number of clouds used in our construction. Since each cloud contains $K+1$ points, the total number of points in $I_\usc$ is $n=c\cdot (K+1)+K$ as we also have $K$ clause points. Therefore, the total complexity of the constructed instance $I_\usc$ is polynomial in $K$ and $N$. Moreover, it is not hard to see that $I_\usc$ can be constructed in polynomial time. We now show the following result.

\begin{lemma}
\label{lem:diskHardnessReduction}
There exists a feasible solution $S$ for $I_\usc$ that covers at least $c\cdot (K+1)$ points uniquely if and only if $I_\sat$ is satisfiable.
\end{lemma}
\begin{proof}
$(\Rightarrow)$ Let $S$ be a feasible solution for $I_\usc$ that covers at least $c\cdot (K+1)$ points. First, by the choice of $c$ and the fact that we have exactly $K$ clause points and any other point belongs to some cloud, we conclude that all the clouds in $I_\usc$ must be covered uniquely by $S$. For each variable $X_i$, where $1\leq i\leq N$, we set the variable $X_i$ to true if its corresponding unit disk is in $S$; otherwise, we set $X_i$ to false. To show that this results in a truth assignment, suppose for a contradiction that there exists a clause $C$ that is not satisfied by this assignment. Take any variable $X\in C$. If $x\in C$ (resp., $\overline{x}\in C$), then the variable $X$ is set to false (resp., true) by the assignment and so the variable unit disk corresponding to $X$ is not in $S$ (resp., is in $S$). Consequently, it follows from the construction that the point clause of $C$ is not covered by the wire arriving from $X$. This means that none of the wires arriving at $C$ will cover its point clause --- this is a contradiction to the feasibility of $S$.

$(\Leftarrow)$ Given a truth assignment for $I_\sat$, we construct a feasible solution $S$ for $I_\usc$ covering at least $c\cdot (K+1)$ points uniquely as follows. For each variable $X_i$ in $I_\sat$, where $1\leq i\leq N$: if $X_i$ is set to true, then we add the corresponding variable unit disk into $S$; otherwise, we add the start disks intersecting this variable unit disk into $S$. Consequently, every other unit disks of the corresponding wires are added into $S$ in such a way that each cloud is covered uniquely by one of the unit disks along the wire. Clearly, all clouds are covered by $S$. Moreover, $S$ also covers all the clause points because the only way to have a clause $C$ satisfied is to have at least one of the wires arriving at $C$ covering the clause point of $C$. Finally, by adding every other unit disk of each wire into $S$, we ensure that all the clouds are covered uniquely by $S$. Since we have $c$ clouds each of which consists of $K+1$ points, we conclude that $S$ is a feasible solution that covers at least $c\cdot (K+1)$ points in $I_\usc$ uniquely.
\end{proof}

By Lemma~\ref{lem:diskHardnessReduction}, we have the following theorem.

\begin{theorem}
\label{thm:uniqueSetCoverDisksHardness}
The \textsc{Unique Disk Set Cover} is \textsc{NP}-hard.
\end{theorem}

\paragraph{\textsc{NP}-Hardness for Unit Squares.} The proof for the \textsc{NP}-hardness of the \textsc{Unique Square Set Cover} problem is almost identical to the one shown above for unit disks. In what follows, we mainly describe the gadgets for unit squares so as then one can verify that the hardness follows.

Starting by the comb-shape form of Knuth and Raghunathan~\cite{KnuthR1992} for an instance of the Planar VR3SAT problem, we replace each variable vertex with a unit square and will then have three (unit-square type) wires connecting the variable unit square to the corresponding clause gadgets. This is again doable even if we have one bend along the wire; see Figure~\ref{fig:gadgetForSquares} for an illustration. Moreover, we can make the three wires arriving at a clause to have a non-empty intersection region in which we insert our single clause point. Finally, the negation pair is constructed in an analogous way. See Figure~\ref{fig:gadgetForSquares}. Note that more unit squares are used (than unit disks) in a wire connecting a variable unit square to its corresponding clauses. However, we can still ensure to have the scaling step to have a polynomial number of unit squares in each wire.

One can verify that the construction is again polynomial in $K$ and $N$, and we can prove a lemma similar to Lemma~\ref{lem:diskHardnessReduction} for this new construction. So, we have the following result.

\begin{figure}[t]
\centering%
\includegraphics[width=.65\textwidth]{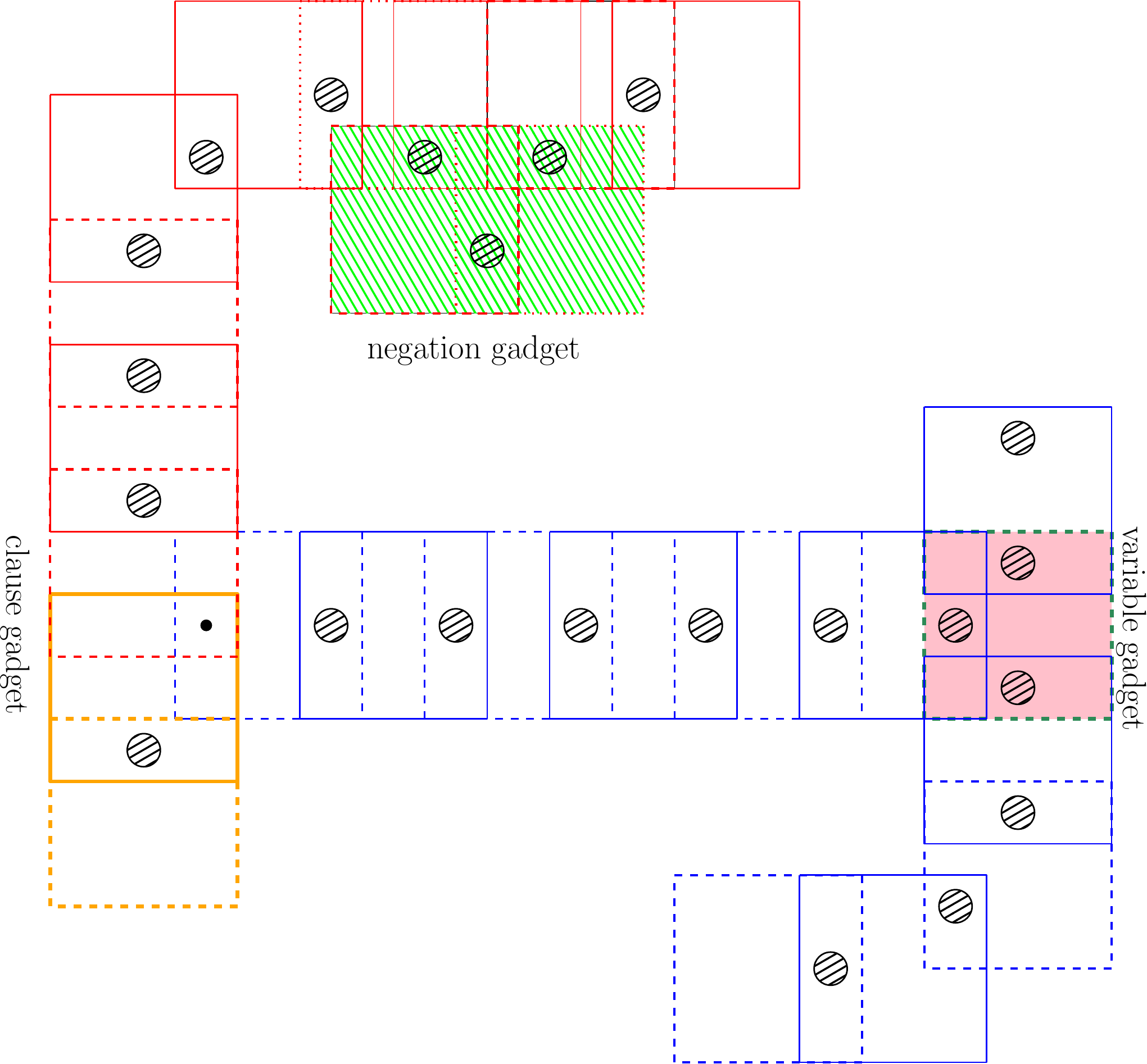}
\caption{{\small An illustration of the gadgets for unit squares. Solid and dashed unit squares alternate for better visibility.}}
\label{fig:gadgetForSquares}%
\end{figure}

\begin{theorem}
\label{thm:unitSquaresHardness}
The \textsc{Unique Square Set Cover} is \textsc{NP}-hard.
\end{theorem}

\section{PTAS}
\label{sec:ptas}
In this section, we give a PTAS for the \textsc{Unique Square Set Cover} problem by applying the mod-one approach of Chan and Hu~\cite{ChanH2015}; we first describe this approach.

Recall the dynamic programming involed in the PTASes developed by Ito et al.~\cite{ItoNOOUUU2016}, and Erlebach and van Leeuwen~\cite{ErlebachL2010}. The dynamic program is essentially based on the line-sweep paradigm by considering points and squares from left to right, and extending the uniquely covered region sequentially. However, adding one square can influence squares that were already chosen and so we need to keep track of the squares that are possibly influenced by a newly-added square. The complication stems mainly from the fact that a new square may influence too many squares. This requires keeping track of the changes on the two separate chains induced by the boundary of the current squares. The mod-one approach avoids this complication by transforming these chains into two chains that are connected at the corner points.

For a set $S = \{s_1,\dots, s_t\}$ of $t$ unit squares containing a common point, where $s_1, \dots, s_t$ are arranged in increasing $x$-order of their centres, the set $S$ is called a \emph{monotone set} if the centres of $s_1, \dots, s_t$ are in increasing or decreasing $y$-order. The boundary of the union of the squares in a monotone set $S$ consists of two monotone chains, called \emph{complementary chains}. In the \emph{mod-one transformation}, a point $(x, y)$ in the plane is mapped to the point $(x \mod 1, y \mod 1)$, where $x \mod 1$ denotes the fractional part of the real number $x$. Applying the mod-one to a monotone set transforms the squares in such a way that the two complementary chains are mapped to two monotone chains that are connected at the corner points. See Figure~\ref{fig:modOne} for an illustration.

\begin{figure}[t]
\centering
\includegraphics[scale=.65]{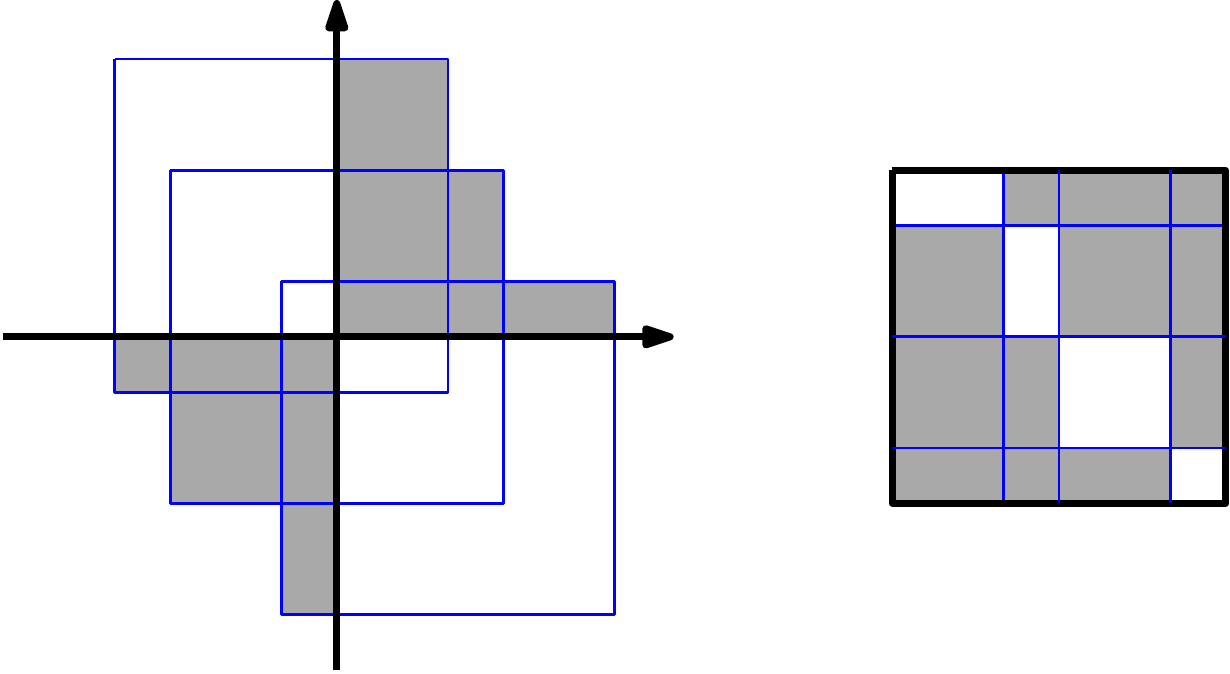}
\caption{{\small A monotone set consisting of three unit squares (left), and the resulting set after applying the mod-one transformation (right).}}
\label{fig:modOne}
\end{figure}

We now show that the mod-one approach provides a PTAS for the \textsc{Unique Square Set Cover} problem. The plan is to first give an exact dynamic programming algorithm for a special variant of the problem, where the points in $P$ are all inside a $k\times k$ square for some constant $k$, and then to apply the shifting strategy of Hochbaum and Maass~\cite{HochbaumM1985} to obtain our PTAS. Note that this follows the framework of~\cite{ChanH2015} with the only difference that we store 6-tuples of unit squares, instead of storing 4-tuples, when solving the dynamic programming.

\begin{lemma}
\label{lem:ourDecomposingOPT}
Let $OPT$ denote an optimal solution for an instance of the \textsc{Unique Square Set Cover} problem in which the set $P$ is inside a $k\times k$ square, for some constant $k$. Then, $OPT$ can be decomposed into $O(k^2)$ monotone sets.
\end{lemma}
\begin{proof}
This lemma is essentially proved in~\cite{ChanH2015}, but for a different problem. The idea is to draw a unit side-length grid over the $k\times k$ square and then show that each unit square appears in the boundary of the union of unit square containing a fixed grid point $p$. By dividing the plane into four quadrants at $p$ and grouping the unit squares containing $p$ based on their contribution to the part of the union boundary in quadrants, $OPT$ is decomposed into $4(k+1)^2$ monotone sets.
\end{proof}

We now give an exact dynamic programming algorithm for the variant of the \textsc{Unique Square Set Cover} problem, where all points of $P$ are inside a $k\times k$ square for some constant $k$.

\begin{theorem}
\label{thm:solvingRestrictedExactly}
For any instance of \textsc{Unique Square Set Cover} problem in which the set $P$ is inside a $k\times k$ square for a constant $k$, the optimal solution can be computed in $O(n\cdot m^{O(k^2)})$ time.
\end{theorem}
\begin{proof}
We describe the dynamic programming algorithm by a state-transition diagram. We store 6-tuples of unit squares in each state. More specifically, a state is defined to consist of \begin{inparaenum}[(i)] \item a vertical sweep line $\ell$ that passes through a corner of an input square, after taking mod 1, and \item $O(k^2)$ 6-tuples of unit squares of the form $(s_{\start}, s_{\prevv}, s_{\prev}, s_{\curr}, s_{\currr}, s_{\en})$ given that $s_{\start}$, $s_{\prevv}$, $s_{\prev}$, $s_{\curr}$, $s_{\currr}$, and $s_{\en}$ are in the increasing order of $x$-coordinate, they form a monotone set, and $\ell$ lies between the corners of $s_{\prev}$ and $s_{\curr} \mod 1$. \end{inparaenum}

Observe that each 6-tuple corresponds to a monotone set $S$: $s_{\start}$ and $s_{\en}$ represent the start and end squares of $S$ and $s_{\prev}$ and $s_{\curr}$ represent the squares of the two complementary chains of $S$, after taking mod 1, that are intersected by the sweep line $\ell$. Moreover, we define $s_{\prevv}$ and $s_{\currr}$ as the predecessor of $s_{\prev}$ and the successor of $s_{\curr}$, respectively. We now define a transition between two states and its cost function. Given a state $A$, we create a transition from $A$ to a new state $B$ as follows. Let $(s_{\start}, s_{\prevv}, s_{\prev}, s_{\curr}, s_{\currr}, s_{\en})$ be the 6-tuple such that the corner point of $s_{\curr}$ has the smallest $x$-coordinate mod 1. First, the new sweep line $\ell'$ is located at the corner of $s_{\curr}$. Next, we replace this 6-tuple by a new 6-tuple $(s_{\start}, s_{\prev}, s_{\curr}, s_{\currr}, s', s_{\en})$, where $s'$ is a unit square that satisfies the conditions of a state. $B$ is now set to this new state with having all other 6-tuples left unchanged. To see the cost of this transition, let $x$ (resp., $y$) be the number of points in $P$ that lie between $\ell$ and $\ell'$, after taking mod 1, and are not covered (resp., are covered uniquely) by the squares from the $O(k^2)$ 6-tuples of unit squares, before taking mod 1. If $x>0$, then we remove this transition from the diagram (since all points must be covered). Otherwise, we set the cost of this transition to $y$.

The problem is then reduced to finding the longest path in this state-transition diagram for which we can appropriately add the start and end states. Since the diagram is a directed acyclic graph, we can construct the diagram and find the longest path in $O(n\cdot m^{O(k^2)})$ time.
\end{proof}

We can now apply the shifting strategy of Hochbaum and Maass~\cite{HochbaumM1985} to obtain our PTAS. The proof of the following theorem is much similar to the one given in~\cite{ChanH2015} and so we omit it here.

\begin{theorem}
\label{thm:shiftingStrategy}
For any fixed constant $\epsilon>0$, there exists a polynomial time $(1+\epsilon)$-approximation algorithm for the \textsc{Unique Square Set Cover} problem.
\end{theorem}

\section{Conclusion}
\label{sec:conclusion}
In this paper, we showed that the \textsc{Unique Disk Set Cover} and \textsc{Unique Square Set Cover} problems are both \textsc{NP}-hard and gave a PTAS for the \textsc{Unique Square Set Cover} problem. Our PTAS is based on the mod-one transformation of Chan and Hu~\cite{ChanH2015}, which does not apply to unit disks. Therefore, giving a PTAS for the \textsc{Unique Disk Set Cover} remains open.

\bibliographystyle{plain}
\bibliography{ref}

\end{document}